\documentclass[a4paper,11pt]{amsart}
\usepackage[active]{srcltx}

\usepackage{graphicx}
\usepackage{amsmath}
\usepackage{amssymb}
\usepackage{hyperref}

\newtheorem{thm}{Theorem}%
\newtheorem{lem}[thm]{Lemma}%
\theoremstyle{remark}
\theoremstyle{plain}

\numberwithin{equation}{section}

\def\QQ{{\mathbb Q}}
\def\RR{{\mathbb R}}

\def\ZZ{{\mathbb Z}}

\def\vecb{{\text{\boldmath$b$}}}

\def\vece{{\text{\boldmath$e$}}}

\def\vecm{{\text{\boldmath$m$}}}

\def\vecq{{\text{\boldmath$q$}}}
\def\vecQ{{\text{\boldmath$Q$}}}
\def\vecp{{\text{\boldmath$p$}}}

\def\vecs{{\text{\boldmath$s$}}}

\def\vecv{{\text{\boldmath$v$}}}
\def\vecV{{\text{\boldmath$V$}}}
\def\vecw{{\text{\boldmath$w$}}}
\def\vecx{{\text{\boldmath$x$}}}

\def\vecy{{\text{\boldmath$y$}}}
\def\vecz{{\text{\boldmath$z$}}}

\def\vecalf{{\text{\boldmath$\alpha$}}}
\def\vecbeta{{\text{\boldmath$\beta$}}}

\def\vecomega{{\text{\boldmath$\omega$}}}

\def\vecxi{{\text{\boldmath$\xi$}}}

\def\vecnull{{\text{\boldmath$0$}}}

\def\scrA{{\mathcal A}}
\def\scrB{{\mathcal B}}

\def\scrD{{\mathcal D}}

\def\scrK{{\mathcal K}}
\def\scrL{{\mathcal L}}

\def\scrP{{\mathcal P}}

\def\scrS{{\mathcal S}}

\def\scrW{{\mathcal W}}

\def\fU{{\mathfrak U}}

\def\e{\mathrm{e}}

\def\intl{{\operatorname{int}}}

\def\GL{\operatorname{GL}}

\def\S{\operatorname{S{}}}

\def\SL{\operatorname{SL}}
\def\ASL{\operatorname{ASL}}

\def\SO{\operatorname{SO}}

\def\T{\operatorname{T{}}}

\def\vol{\operatorname{vol}}

\def\ASLASL{\ASL(d,\ZZ)\backslash\ASL(d,\RR)}
\def\ASLZ{\ASL(d,\ZZ)}
\def\ASLR{\ASL(d,\RR)}

\def\SLZ{\SL(d,\ZZ)}
\def\SLR{\SL(d,\RR)}

\def\trans{\,^\mathrm{t}\!}

\def\bs{\backslash}

\def\nbar{\overline{n}}

\def\Onder#1#2#3#4#5{#1 \setbox0=\hbox{$#1$}\setbox1=\hbox{$#2$}
       \dimen0=.5\wd0 \dimen1=\dimen0 \dimen2=\dp0 \dimen3=\dimen2
       \advance\dimen0 by .5\wd1 \advance\dimen0 by -#4
       \advance\dimen1 by -.5\wd1 \advance\dimen1 by -#4
       \advance\dimen2 by -#3 \advance\dimen2 by \ht1
       \advance\dimen2 by 0.3ex \advance\dimen3 by #5
        \kern-\dimen0\raisebox{-\dimen2}[0ex][\dimen3]{\box1}
       \kern\dimen1}

\newcommand{\asl}{\mathfrak{asl}}
\newcommand{\lsl}{\mathfrak{sl}}
\newcommand{\ig}{\mathfrak{g}}

\newcommand{\il}{\mathfrak{l}}
\newcommand{\ik}{\mathfrak{k}}

\newcommand{\Q}{\mathbb{Q}}
\newcommand{\R}{\mathbb{R}}
\newcommand{\Z}{\mathbb{Z}}

\newcommand{\HS}{{{\S'_1}^{d-1}}}

\newcommand{\col}{\: : \:}

\newcommand{\bn}{\mathbf{0}}

\title{Power-law distributions for the free path length in Lorentz gases}
\author{Jens Marklof}
\author{Andreas Str\"ombergsson}
\address{School of Mathematics, University of Bristol,
Bristol BS8 1TW, U.K.\newline
\rule[0ex]{0ex}{0ex} \hspace{8pt}{\tt j.marklof@bristol.ac.uk}}
\address{Department of Mathematics, Box 480, Uppsala University,
SE-75106 Uppsala, Sweden\newline
\rule[0ex]{0ex}{0ex} \hspace{8pt}{\tt astrombe@math.uu.se}}
\date{\today}
\thanks{The research leading to these results has received funding from the European Research Council under the European Union's Seventh Framework Programme (FP/2007-2013) / ERC Grant Agreement n. 291147. 
J.M.\ is furthermore supported by a Royal Society Wolfson Research Merit Award, and A.S.\ is a Royal Swedish Academy of Sciences Research Fellow supported by a grant from the Knut and Alice Wallenberg Foundation.}

\begin{document}

\begin{abstract}
It is well known that, in the Boltzmann-Grad limit, the distribution of the free path length in the Lorentz gas with disordered scatterer configuration has an exponential density. If, on the other hand, the scatterers are located at the vertices of a Euclidean lattice, the density has a power-law tail proportional to $\xi^{-3}$. In the present paper we construct scatterer configurations whose free path lengths have a distribution with tail $\xi^{-N-2}$ for any positive integer $N$. We also discuss the properties of the random flight process that describes the Lorentz gas in the Boltzmann-Grad limit. The convergence of the distribution of the free path length follows from equidistribution of large spheres in products of certain homogeneous spaces, which in turn is a consequence of Ratner's measure classification theorem.
\end{abstract}

\maketitle

\section{Introduction}

The Lorentz gas \cite{Lorentz05} describes the dynamics of non-interacting point particles in an array of fixed spherical scatterers of radius $\rho$, centered at the elements of a point set $\scrP\subset\RR^d$, with $d\geq2$. We assume that $\scrP$ has unit density, i.e., for any bounded $\scrD\subset\RR^d$ with boundary of measure zero,
\begin{equation}\label{density000}
\lim_{T\to\infty} \frac{ \#(\scrP\cap T \scrD)}{T^d} = \vol(\scrD) .
\end{equation}
Each particle travels with constant velocity along straight lines until it enters a scatterer where it is deflected, e.g. by elastic reflection (as in the classic setting of the Lorentz gas) or by the force of a spherically symmetric potential. 
We denote the position and velocity at time $t$ by $\vecq(t)$ and $\vecv(t)$. Since the particle speed outside the scatterers is a constant of motion we may assume without loss of generality $\|\vecv(t)\|=1$. The dynamics thus takes place in the unit tangent bundle $\T^1(\scrK_\rho)$
where $\scrK_\rho\subset\RR^d$ is the complement of the set $\scrB^d_\rho + \scrP$; $\scrB^d_\rho$ denotes the open ball of radius $\rho$, centered at the origin. We parametrize $\T^1(\scrK_\rho)$ by $(\vecq,\vecv)\in\scrK_\rho\times\S_1^{d-1}$, where we use the convention that for $\vecq\in\partial\scrK_\rho$ the vector $\vecv$ points away from the scatterer (so that $\vecv$ describes the velocity {\em after} the collision). 
The Liouville measure on $\T^1(\scrK_\rho)$ is 
\begin{equation} \label{LIOUVILLEDEF}
	d\nu(\vecq,\vecv)=d\!\vol_{\RR^d}(\vecq)\, d\!\vol_{\S_1^{d-1}}(\vecv)
\end{equation}
where $\vol_{\RR^d}$ and $\vol_{\S_1^{d-1}}$ refer to the Lebesgue measures on $\RR^d$ %
and $\S_1^{d-1}$, respectively. 

The first collision time with respect to the initial condition $(\vecq,\vecv)\in\T^1(\scrK_\rho)$ is 
\begin{equation} \label{TAU1DEF0}
	\tau_1(\vecq,\vecv;\rho) = \inf\{ t>0 : \vecq+t\vecv \notin\scrK_\rho \}. 
\end{equation}
Since all particles are moving with unit speed, we may also refer to $\tau_1(\vecq,\vecv;\rho)$ as the free path length. The distribution of free path lengths in the limit of small scatterer density (Boltzmann-Grad limit) has been studied extensively when $\scrP$ is a fixed realisation of a random point process (such as a spatial Poisson process) \cite{Boldrighini83,Gallavotti69,Polya18,Spohn78} and when $\scrP$ is a Euclidean lattice \cite{Boca03,Boca07,Bourgain98,Caglioti03,Dahlqvist97,Dettmann12,Golse00,partI,Nandori12,Polya18} (cf.~also the recent studies of the free path length in the honeycomb lattice \cite{Boca09,Boca10} and quasicrystals \cite{Wennberg12,quasi}). 
In both cases, the limit distribution for the free path length between consecutive collisions exists in the small scatterer limit: For a random scatterer configuration, the probability that $\tau_1>\rho^{-(d-1)} \xi$ has the limiting density (as $\rho\to 0$)
\begin{equation}\label{expony}
\overline\Phi_{\vecnull,\scrP}(\xi) = \overline\sigma\, \e^{-\overline\sigma \xi}
\end{equation}
where $\overline\sigma$ is the volume of the $(d-1)$ dimensional unit ball (this represents the total cross section of a spherical scatterer in units of the radius). If the scatterer configuration is given by a Euclidean lattice $\scrP=\scrL$, an explicit formula for the limit distribution of free path length is only known in dimension $d=2$ \cite{Boca03,Boca07,Dahlqvist97}; in higher dimension we have the tail estimates 
$\overline{\Phi}_{\vecnull,\scrL}(\xi)\sim A_d \xi^{-3}$ for large $\xi$ (the value of $A_d$ is given in Section \ref{sec:Asymptotic}), and $\overline{\Phi}_{\vecnull,\scrL}(\xi)=\frac{\overline\sigma}{\zeta(d)}+O(\xi)$, for $\xi\to 0$ \cite{partIV}.
Note that the heavy tail at infinity implies that $\overline{\Phi}_{\vecnull,\scrL}(\xi)$ has no second moment. A further non-trivial observation of \cite{partI} is that the limit distributions are independent of the choice of lattice $\scrL$ (the covolume of $\scrL$ is assumed to be one). We will therefore set in the following $\overline{\Phi}_{\vecnull}(\xi):=\overline{\Phi}_{\vecnull,\scrL}(\xi)$.

In this study we assume that the scattering configuration is given by a finite union of distinct Euclidean lattices. That is,  
\begin{equation}
\scrP = \bigcup_{i=1}^N \scrL_i
\end{equation}
where $\scrL_i$ are (possibly shifted) Euclidean lattices of covolume $\nbar_i^{-1}$.
We will assume that the lattices are pairwise incommensurable 
(we will give a precise definition of this in Section \ref{sec:Free} below);
this ensures among other things that each intersection $\scrL_i\cap\scrL_j$ ($i\neq j$) is contained in some 
affine subspace of dimension less than $n$, and thus the density of $\scrP$ is $\nbar_1+\ldots+\nbar_N$. 
We impose the normalizing condition $\nbar_1+\ldots+\nbar_N=1$. 

We will prove that %
the distribution of free path lengths has a limit density with tails
\begin{align}\label{PHIBARXILARGETHMRES1111}
\overline{\Phi}_{\vecnull,\scrP}(\xi)\sim
C_\infty \,\xi^{-(N+2)}
\qquad \text{as } \: \xi\to\infty,
\end{align}
and
\begin{equation}
\overline{\Phi}_{\vecnull,\scrP}(\xi)=C_0+O(\xi), \qquad \text{as } \: \xi\to 0;
\end{equation}
see Theorems \ref{thm:tail} and \ref{thm:tail0} in Section \ref{sec:Asymptotic} below for explicit formulas for the constants $C_0$ and $C_\infty$.

In the Boltzmann-Grad limit, the Lorentz process in fact may converge to a random flight process; this has been proved in the case of random $\scrP$ in \cite{Gallavotti69,Spohn78,Boldrighini83}, and in the case of Euclidean lattices $\scrP=\scrL$ in \cite{partII}. In the random setting, this limiting process is governed by the linear Boltzmann equation as originally suggested by Lorentz \cite{Lorentz05}. In the periodic setting, the linear Boltzmann equation has to be replaced by a more general transport equation. We will argue that the same applies in the setting studied in this paper (Section \ref{sec:Boltzmann}), and establish a limit theorem for the transition kernel that characterises the transport equation (Section \ref{sec:Transition}). The key technical ingredient in our proofs (which follow closely the strategy developed in \cite{partI}) is an application of Ratner's measure classification theorem to establish equidistribution in products of certain homogeneous spaces (Section \ref{sec:Equidistribution}).

\section{Free path length}\label{sec:Free}

We will begin by describing our results on the distribution of free path lengths for initial data of the form $(\vecq+\rho \vecbeta(\vecv),\vecv)$, where $\vecv\in\S_1^{d-1}$ is random, $\vecq\in\RR^d$ is fixed and $\vecbeta:\S_1^{d-1}\to\RR^{d}$ is some fixed continuous function. For $\vecq\notin\scrP$, the free path length $\tau_1(\vecq+\rho\vecbeta(\vecv),\vecv;\rho)$ is evidently well defined for $\rho$ sufficiently small. If $\vecq\in\scrP$, we assume in the following that $\vecbeta$ is chosen so that
the ray $\vecbeta(\vecv)+\R_{\geq 0}$ lies completely outside the ball $\scrB_1^d$ for all $\vecv\in\S^{d-1}_1$
(thus $\rho\vecbeta(\vecv)+\R_{\geq 0}$ lies outside $\scrB_\rho^d$ for all $\rho>0$);
this is to avoid any initial condition where the particle starts inside the scatterer. 

We let $\scrS$ be the commensurator of $\SL(d,\Z)$ in $\SL(d,\R)$. Thus
\begin{align*}
\scrS=
\{(\det T)^{-1/d}T\col T\in\GL(d,\Q),\:\det T>0\}.
\end{align*}
Cf.\ \cite[Thm.\ 2]{borel}, as well as \cite[Sec.\ 7.3]{studenmund}.
We say that the matrices {\em $M_1,\ldots,M_N\in\SLR$ are pairwise incommensurable} if $M_i M_j^{-1}\notin\scrS$ for all $i\neq j$. A simple example is 
\begin{equation}
M_i = \zeta^{-i/d} \begin{pmatrix} \zeta^{i} & 0 \\ 0 & 1_{d-1} \end{pmatrix}, \qquad i=1,\ldots,N,
\end{equation}
where $\zeta$ is any positive number such that $\zeta,\zeta^2,\ldots,\zeta^{N-1}\notin\Q$.

\begin{thm}\label{freeThm}
Fix $N$ affine lattices $\scrL_i=\nbar_i^{-1/d} (\Z^d+\vecomega_i) M_i$, $i=1,\ldots,N$, with $\vecomega_i\in\R^d$, $M_i\in\SL(d,\R)$ pairwise incommensurable, and $\nbar_i>0$ such that $\nbar_1+\ldots+\nbar_N=1$.
Let $\vecq\in\R^d$ and set $\vecalf_i=\vecomega_i-\nbar_i^{1/d} \vecq M_i^{-1}$.
Then, for every $\xi\geq 0$,
\begin{equation}\label{FPL}
\lim_{\rho\to 0} \lambda(\{ \vecv\in\S_1^{d-1} \col  \rho^{d-1} \tau_1(\vecq+\rho\vecbeta(\vecv),\vecv;\rho)\geq \xi \})
= \int_{\xi}^\infty \Phi_{\vecalf_1,\ldots,\vecalf_N,\vecbeta}(\xi')\, d\xi'  
\end{equation}
with
\begin{equation}\label{limid}
\Phi_{\vecalf_1,\ldots,\vecalf_N,\vecbeta}(\xi)= -\frac{d}{d\xi} 
\prod_{i=1}^N\int_{\nbar_i \xi}^\infty \Phi_{\vecalf_i,\vecbeta}(\xi')\, d\xi',
\end{equation}
where $\Phi_{\vecalf,\vecbeta}(\xi)$ is the continuous probability density on $\R_{> 0}$ defined in \cite[Eq.~(4.6)]{partI}.
\end{thm}

Let us fix a map $K:\S_1^{d-1}\to\SO(d)$ such that
$\vecv K(\vecv)=\vece_1$ for all $\vecv\in\S_1^{d-1}$;
we assume that $K$ is smooth when restricted to $\S_1^{d-1}$ minus
one point (see \cite[footnote 3, p.~1968]{partI} for an explicit construction). We show in \cite{partI} that for $\vecalf\notin\Q^d$, the density $\Phi(\xi):=\Phi_{\vecalf,\vecbeta}(\xi)$ is independent of $\vecalf$ and $\vecbeta$. If $\vecalf\in\Z^d$, then $\Phi_{\vecalf,\vecbeta}(\xi)=\Phi_{\vecnull,\vecbeta}(\xi)$ can be expressed as 
\begin{equation}
\Phi_{\vecnull,\vecbeta}(\xi) = \int_{\S_1^{d-1}} \Phi_\vecnull(\xi,(\vecbeta(\vecv) K(\vecv))_\perp) d\lambda(\vecv),
\end{equation}
where $\vecx_\perp$ denotes the orthogonal projection of $\vecx\in\R^d$ onto $\{0\}\times\RR^{d-1}$, 
and $\Phi_\vecnull(\xi,\vecz)$ is a probability density on $\R_{>0}$ for any fixed $\vecz\in\{0\}\times\R^{d-1}$. We will call $\vecs(\vecv):=(\vecbeta(\vecv) K(\vecv))_\perp\in\scrB_1^{d-1}$ the exit parameter.

Let us discuss two special examples of the limit density \eqref{limid}. In the first, we are shooting off from a scatterer centered at a lattice point $\vecq\in\scrL_j$ (i.e., $\vecalf_j\in\Z^d$) for some $j$, such that $\vecalf_i\notin\Q^d$ for each $i\neq j$. 
We remark that this condition, with $j$ depending on $\vecq$,
is satisfied for all points in $\vecq\in\scrP$ outside a finite union of affine subspaces of dimension less than $n$;
this is because of our assumption about pairwise incommensurability.
In particular the condition holds for asymptotically all points $\vecq\in\scrP$.
For such a starting point $\vecq$, the right hand side of \eqref{limid} is
\begin{equation}\label{limid2}
\Phi_{\vecalf_1,\ldots,\vecalf_N,\vecbeta}(\xi)= 
\int_{\S_1^{d-1}}\Phi_{\vecnull}^{(j)}(\xi, (\vecbeta(\vecv) K(\vecv))_\perp)\,d\lambda(\vecv)
\end{equation}
where 
\begin{equation}\label{Phi0P}
\Phi_{\vecnull}^{(j)}(\xi,\vecz)
= -\frac{d}{d\xi}  \bigg( \int_{\nbar_j \xi}^\infty \Phi_{\vecnull}(\xi',\vecz) d\xi' \; \prod_{\substack{i=1 \\ i\neq j}}^N \int_{\nbar_i \xi}^\infty \Phi(\xi') d\xi' \bigg).
\end{equation}
The distribution of the free path length between consecutive collisions is
\begin{equation}\label{PhiP0}
\overline\Phi_{\vecnull,\scrP}(\xi) = \frac{1}{\overline\sigma} \sum_{j=1}^N \nbar_j \, \int_{\scrB_1^{d-1}} \Phi_{\vecnull}^{(j)}(\xi,\vecz) \, d\vecz.
\end{equation}

In the second example, we launch a particle from a generic point $\vecq$, by which we mean here  $\vecalf_i\notin\Q^d$ for all $i$. Now the right hand side of \eqref{FPL} yields 
\begin{equation}\label{PhiP}
\Phi_{\vecalf_1,\ldots,\vecalf_N,\vecbeta}(\xi) =\Phi_{\scrP}(\xi)
:= -\frac{d}{d\xi} \prod_{i=1}^N \int_{\nbar_i \xi}^\infty \Phi(\xi') d\xi' .
\end{equation}
This density is, as in the single-lattice setting, independent of $\vecbeta$.
We recall from \cite{partI},
\begin{equation}
\overline\Phi_{\vecnull}(\xi)=\frac1{\overline\sigma}\int_{\scrB_1^{d-1}}\Phi_\bn(\xi,\vecz)\,d\vecz 
=  - \frac{1}{\overline\sigma}\; \frac{d}{d\xi} \Phi(\xi) .
\end{equation}
With the above relations, this implies the following relation between the distribution of the free path length between consecutive collisions and the distribution of the distance to the first scatter from a generic point,
\begin{equation}
\overline\Phi_{\vecnull,\scrP}(\xi) =  - \frac{1}{\overline\sigma}\; \frac{d}{d\xi} \Phi_{\scrP}(\xi) .
\end{equation}

\section{Asymptotic tails}\label{sec:Asymptotic}

For the distribution of the free path length in a single lattice, we have proved in \cite{partIV} that for $\xi$ small
\begin{align}\label{0asymp}
\overline{\Phi}_\bn(\xi)=\frac{\overline\sigma}{\zeta(d)}+O(\xi);
\qquad
\Phi(\xi)=\overline\sigma-\frac{\overline\sigma^2}{\zeta(d)}\xi+O(\xi^2) .
\end{align}
The tail asymptics for large $\xi$ are
\begin{align}\label{PHIXILARGETHMRES1}
\Phi(\xi)=
\frac{\overline\sigma}{2}\, A_d\, \xi^{-2}
+O\bigl(\xi^{-2-\frac 2d}\bigr)
\qquad \text{as } \: \xi\to\infty,
\end{align}
\begin{align}\label{PHIBARXILARGETHMRES1}
\overline{\Phi}_\bn(\xi)=
A_d\, \xi^{-3}
+O\bigl(\xi^{-3-\frac 2d}\log\xi\bigr)
\qquad \text{as } \: \xi\to\infty,
\end{align}
where 
\begin{equation}
A_d=\frac{2^{2-d}}{d(d+1)\zeta(d)}
\end{equation}
and the $\log\xi$ may be replaced by $1$ when $d\neq 3$.
These relations imply, by an elementary calculation, the following asymptotic estimates for the free path lengths \eqref{Phi0P} and \eqref{PhiP}:

\begin{thm}\label{thm:tail}
For $\xi\to\infty$,
\begin{equation}
\overline\Phi_{\vecnull,\scrP}(\xi) = 
\frac{N(N+1) A_d^N \overline\sigma^{N-1}}{2^N \nbar_1\cdots\nbar_N} \; \xi^{-(N+2)}
\bigl(1+O(\xi^{-\frac2d}\log\xi)\bigr),
\end{equation}
\begin{equation}
\Phi_{\scrP}(\xi) = 
\frac{N A_d^N \overline\sigma^{N}}{2^N \nbar_1\cdots\nbar_N} \; \xi^{-(N+1)}
\bigl(1+O(\xi^{-\frac2d})\bigr),
\end{equation}
where, in the first expression, $\log\xi$ may be replaced by $1$ when $d\neq 3$.
\end{thm}

\begin{thm}\label{thm:tail0}
For $\xi\geq 0$,
\begin{equation}
\overline\Phi_{\vecnull,\scrP}(\xi) = \overline\sigma \bigg( 1-\bigg(1-\frac{1}{\zeta(d)}\bigg)  \sum_{j=1}^N \nbar_j^2 \bigg) +O(\xi) ,
\end{equation}
\begin{equation}
\Phi_{\scrP}(\xi) = \overline\sigma -  \overline\sigma^2  \bigg( 1-\bigg(1-\frac{1}{\zeta(d)}\bigg)  \sum_{j=1}^N \nbar_j^2 \bigg) \; \xi +O(\xi^2) .
\end{equation}

\end{thm}

\section{A macroscopic transport equation}\label{sec:Boltzmann}

As the mean free path length scales like $\rho^{-(d-1)}$, i.e., like the inverse of the total scattering cross section of an individual scatterer, we rescale all length units by introducing the macroscopic coordinates
\begin{equation}\label{macQV}
	\big(\vecQ(t),\vecV(t)\big) = \big(\rho^{d-1} \vecq(\rho^{-(d-1)} t),\vecv(\rho^{-(d-1)} t)\big) .
\end{equation}
This rescaling of length and time is called the {\em Boltzmann-Grad scaling}, and the corresponding limit as $\rho\to 0$ is called the {\em Boltzmann-Grad limit}. We define the macroscopic particle flow by $F_\rho^t(\vecQ,\vecV):=(\vecQ(t),\vecV(t))$ and consider the evolution of an initial macroscopic particle density $f_0(\vecQ,\vecV)$ which, at time $t>0$, reads $f_t^{(\rho)}(\vecQ,\vecV)=f_0(F_\rho^{-t}(\vecQ,\vecV))$.

The work of Gallavotti \cite{Gallavotti69}, Spohn \cite{Spohn78} and Boldrighini, Bunimovich and Sinai \cite{Boldrighini83} shows that, if the scatterer configuration $\scrP$ is random (e.g.~given by a typical realisation of a Poisson process studied in \cite{Boldrighini83}), then $f_t^{(\rho)}(\vecQ,\vecV)$ converges in the Boltzmann-Grad limit (in a weak sense) to a solution of the linear Boltzmann equation
\begin{equation}
		\bigg[ \frac{\partial}{\partial t} + \vecV\cdot\nabla_\vecQ \bigg] f_t(\vecQ,\vecV)= \int_{\S_1^{d-1}}  \big[ f_t(\vecQ,\vecV_0)-f_t(\vecQ,\vecV)\big] \sigma(\vecV_0,\vecV) \, d\vecV_0,
\end{equation}
as predicted by Lorentz \cite{Lorentz05}. Here $\sigma(\vecV_0,\vecV)$ denotes the differential cross section of the scatterer.  As mentioned in the introduction, the distribution of free path lengths between consecutive collisions is in this case given by the exponential distribution \eqref{expony}. 

In the case when the scatterer configuration is a single Euclidean lattice, we have shown in \cite{partII} that, for a general class of scattering maps (which include elastic reflection and muffin-tin Coulomb potentials), the particle density $f_t^{(\rho)}(\vecQ,\vecV)$ also converges as $\rho\to 0$, but the limiting density $f_t(\vecQ,\vecV)$ does not satisfy a linear transport equation; cf. also \cite{GolseToulouse2008}, \cite{CagliotiGolse}. To obtain a transport equation, consider a density $f(\vecQ,\vecV,\xi,\vecV_+)$ on an extended phase space, where $\xi\in\RR_+$ is the flight time until the next collision, and $\vecV_+\in\S_1^{d-1}$ is the velocity after the next collision. The density $f_t(\vecQ,\vecV,\xi,\vecV_+)$ satisfies a {\em generalized} linear Boltzmann equation
\begin{equation}\label{glB}
	\bigg[ \frac{\partial}{\partial t} + \vecV\cdot\nabla_\vecQ - \frac{\partial}{\partial\xi} \bigg] f_t(\vecQ,\vecV,\xi,\vecV_+) 
	= \int_{\S_1^{d-1}}  f_t(\vecQ,\vecV_0,0,\vecV) \,
p_{\vecnull}(\vecV_0,\vecV,\xi,\vecV_+) \,
d\vecV_0 .
\end{equation} 
with collision kernel
\begin{equation}
	p_{\vecnull}(\vecV_0,\vecV,\xi,\vecV_+) 
	=\sigma(\vecV,\vecV_+)\, \Phi_\vecnull\big(\xi,-\vecb(\vecV,\vecV_+),
	s(\vecV,\vecV_0)\big)
\end{equation}
where $\sigma(\vecV,\vecV_+)$ is the differential cross section and $\Phi_\vecnull(\xi,-\vecb,\vecs)$ the transition kernel to exit with parameter $\vecs$ and hit the next scatterer at time $\xi$ with impact parameter $\vecb$.
If we choose the initial condition
\begin{equation}
	\lim_{t\to 0}f_t(\vecQ,\vecV,\xi,\vecV_+) = f_0(\vecQ,\vecV)\, p(\vecV,\xi,\vecV_+) 
\end{equation}
with
	\begin{equation}
	p(\vecV,\xi,\vecV_+) := \int_\xi^\infty \int_{\S_1^{d-1}} \sigma(\vecV_0,\vecV) p_{\vecnull}(\vecV_0,\vecV,\xi',\vecV_+)\, d\vecV_0\, d\xi',
	\end{equation}
then 
\begin{equation}
f_t(\vecQ,\vecV)=\int_0^\infty \int_{\S_1^{d-1}} f(\vecQ,\vecV,\xi,\vecV_+)\, d\vecV_+\, d\xi 
\end{equation}
is the weak limit of $f_t^{(\rho)}(\vecQ,\vecV)$ as $\rho\to 0$.	
Note that $p(\vecV,\xi,\vecV_+)$ is a stationary solution of the generalized linear Boltzmann equation.

Let us now introduce the generalised linear Boltzmann equation that describes the Boltzmann-Grad limit of $f_t^{(\rho)}(\vecQ,\vecV)$ in the case of a finite union of Euclidean lattices $\scrL_i$ (as in Theorem \ref{freeThm}). We will limit our discussion to the convergence of the transition kernel in analogy with our results in the periodic setting in \cite{partI}. The weak convergence of $f_t^{(\rho)}(\vecQ,\vecV)$ to a solution of the transport equation would require a number of additional technical estimates that are beyond the scope of this short note (cf.~\cite{partII} for the periodic setting).

The particle density will now not only depend on position, velocity, time to the next hit and velocity thereafter but also on the index $j$ which indicates the lattice $\scrL_j$ that is involved in the next collision. (By assumption, the points in $\scrP$ that belong to two or more lattices have density zero, and will not be relevant for the limiting process.) The generalised linear Boltzmann equation reads thus
\begin{multline}\label{glB22}
	\bigg[ \frac{\partial}{\partial t} + \vecV\cdot\nabla_\vecQ - \frac{\partial}{\partial\xi} \bigg] f_t^{(j)}(\vecQ,\vecV,\xi,\vecV_+) \\
	= \sum_{i=1}^N \int_{\S_1^{d-1}}  f_t^{(i)}(\vecQ,\vecV_0,0,\vecV) 
p_{\vecnull}^{(i\to j)}(\vecV_0,\vecV,\xi,\vecV_+) \,
d\vecV_0 .
\end{multline}
The original particle density is recovered from the relation 
\begin{equation}
f_t(\vecQ,\vecV)=\sum_{i=1}^N \int_0^\infty \int_{\S_1^{d-1}} f^{(i)}(\vecQ,\vecV,\xi,\vecV_+)\, d\vecV_+\, d\xi .
\end{equation}
The right hand side of \eqref{glB22} involves a new collision kernel $p_{\vecnull}^{(i\to j)}(\vecV_0,\vecV,\xi,\vecV_+)$, given by
\begin{equation}
	p_{\vecnull}^{(i\to j)}(\vecV_0,\vecV,\xi,\vecV_+) 
	=\sigma(\vecV,\vecV_+)\, \Phi_\vecnull^{(i\to j)}\big(\xi,-\vecb(\vecV,\vecV_+),
	\vecs(\vecV,\vecV_0)\big)
\end{equation}
where $\sigma(\vecV,\vecV_+)$ is the differential cross section and $\Phi_\vecnull^{(i\to j)}(\xi,-\vecb,\vecs)$ the transition probability density to hit the next scatterer in $\scrL_j$ at time $\xi$ with impact parameter $\vecb$, given that the particle exits a previous scatterer in $\scrL_i$ with parameter $\vecs$. So $\Phi_\vecnull^{(i\to j)}(\xi,\vecw,\vecz)$ is a conditional probability density with
\begin{equation}\label{conditor} %
\sum_{j=1}^N \int_0^\infty \int_{\scrB_1^{d-1}} \Phi_\vecnull^{(i\to j)}(\xi,\vecw,\vecz) \,d\vecw\, d\xi =1 
\end{equation}
for all $\vecz$.

The stationary solution of \eqref{glB22} is $f_t^{(j)}(\vecQ,\vecV,\xi,\vecV_+) =p^{(j)}(\vecV,\xi,\vecV_+)$ where
\begin{equation}
	p^{(j)}(\vecV,\xi,\vecV_+) := \sum_{i=1}^N \nbar_i \int_\xi^\infty \int_{\S_1^{d-1}} \sigma(\vecV_0,\vecV) p_{\vecnull}^{(i\to j)}(\vecV_0,\vecV,\xi',\vecV_+)\, d\vecV_0\, d\xi' .
\end{equation}
In particular we have
\begin{equation}
	p^{(j)}(\vecV,0,\vecV_+) = \nbar_j  \sigma(\vecV,\vecV_+) .
\end{equation}
This follows from eq.~\eqref{pj} below.

Let us now discuss the convergence of the microscopic transition probabilities to the transition kernel $\Phi_\vecnull^{(i\to j)}\big(\xi,\vecw,\vecz)$, and provide explicit formulas in terms of the single-lattice distributions.

\section{Convergence of the transition kernel}\label{sec:Transition}

We are now interested in the joint distribution of the free path length (considered in Section~\ref{sec:Free}), and the precise location {\em on} the scatterer where the particle hits.

Given initial data $(\vecq,\vecv)$, we denote the position of impact on the first scatterer by
\begin{equation}
	\vecq_1(\vecq,\vecv;\rho) := \vecq+\tau_1(\vecq,\vecv;\rho) \vecv .
\end{equation}
We define the function $h_1(\vecq,\vecv;\rho)=j$ if the first scatterer hit is centered at a point $\vecy$ in the lattice $\scrL_j$. (If two ore more scatterers overlap and are hit simultaneously, we choose the scatterer belonging to the lattice with the smallest index to make $h_1(\vecq,\vecv;\rho)$ well defined; this scenario is a probability zero event.) Given the scatterer location $\vecy\in\scrL_{h_1(\vecq,\vecv;\rho)}$, we have
$\vecq_1(\vecq,\vecv;\rho)\in \S_\rho^{d-1} + \vecy$ and therefore there is a unique point
$\vecw_1(\vecq,\vecv;\rho)\in \S_1^{d-1}$ such that
$\vecq_1(\vecq,\vecv;\rho)=\rho \vecw_1(\vecq,\vecv;\rho)+\vecy$. 
It is evident that $-\vecw_1(\vecq,\vecv;\rho) K(\vecv)\in \HS$, with the hemisphere $\HS=\{\vecv=(v_1,\ldots,v_d)\in\S_1^{d-1} \col v_1>0\}$. The impact parameter of the first collision is $\vecb=(\vecw_1(\vecq,\vecv;\rho) K(\vecv))_\perp$.

As in Section \ref{sec:Free}, we will use the initial data $(\vecq+\rho\vecbeta(\vecv),\vecv)$ for fixed $\vecq$ and $\vecbeta$, and use the shorthand $\tau_1=\tau_1(\vecq+\rho\vecbeta(\vecv),\vecv;\rho)$, $\vecw_1=\vecw_1(\vecq+\rho\vecbeta(\vecv),\vecv;\rho)$ and $h_1=h_1(\vecq+\rho\vecbeta(\vecv),\vecv;\rho)$.

\begin{thm}\label{exactpos1}
Fix $N$ affine lattices $\scrL_i=\nbar_i^{-1/d} (\Z^d+\vecomega_i) M_i$, $i=1,\ldots,N$, with $\vecomega_i\in\R^d$, $M_i\in\SL(d,\R)$ pairwise incommensurable, and $\nbar_i>0$ such that $\nbar_1+\ldots+\nbar_N=1$.
Let $\vecq\in\R^d$ and set $\vecalf_i=\vecomega_i-\nbar_i^{1/d} \vecq M_i^{-1}$.
Then for any Borel probability measure 
$\lambda$ on $\S_1^{d-1}$ absolutely 
continuous with respect to $\vol_{\S_1^{d-1}}$, any
subset $\fU\subset\HS$ with $\vol_{\S_1^{d-1}}(\partial\fU)=0$, 
and any $0\leq a< b$, we have
\begin{multline} \label{exactpos1eq}
\lim_{\rho\to 0}  \lambda\bigl(\bigl\{ \vecv\in\S_1^{d-1} \col 
\rho^{d-1} \tau_1\in [a,b), \:  
-\vecw_1K(\vecv)\in\fU,\: h_1=j\bigr\}\bigr) \\
=\int_{a}^{b} \int_{\fU_\perp} \int_{\S_1^{d-1}} 
\Phi_{\vecalf_1,\ldots,\vecalf_N}^{(j)}\bigl(\xi,\vecw,(\vecbeta(\vecv)K(\vecv))_\perp\bigr) 
\, d\lambda(\vecv)\, d\vecw \, d\xi,
\end{multline}
where
\begin{equation}
\Phi_{\vecalf_1,\ldots,\vecalf_N}^{(j)}(\xi,\vecw,\vecz) = \nbar_j
\Phi_{\vecalf_j}(\nbar_j\xi,\vecw,\vecz) \prod_{\substack{i=1 \\ i\neq j}}^N
\int_{\nbar_i \xi}^\infty \int_{\scrB_1^{d-1}} \Phi_{\vecalf_i}(\xi',\vecw',\vecz)\, d\vecw'\,d\xi' ,
\end{equation}
and  $\Phi_{\vecalf}(\xi,\vecw,\vecz)$ is the transition kernel for a single Euclidean lattice of covolume one.
\end{thm}

We refer the reader to \cite{partI} for a detailed study of $\Phi_{\vecalf}(\xi,\vecw,\vecz)$. 
In particular we note in \cite[Remark 4.5]{partI} that for $\vecalf\notin\Q^d$,
the kernel $\Phi(\xi,\vecw):=\Phi_{\vecalf}(\xi,\vecw,\vecz)$ is independent of $\vecalf$ and $\vecz$.
An explicit formula of the transition kernel $\Phi_{\vecalf}(\xi,\vecw,\vecz)$ in dimension $d=2$ is derived in \cite{partIII} (cf.~also \cite{CagliotiGolse,Bykovskii09}); for asymptotics in higher dimensions $d\geq 2$ see \cite{partIV}.

As discussed earlier in the case of the free path length, a particularly relevant case is when $\vecq\in\scrL_k$ for some $k$ (i.e., $\vecalf_k\in\ZZ^d$) and $\vecq$ is generic with respect to the remaining lattices, i.e., $\vecalf_{i}\notin\QQ^d$ for $i\neq k$. In this case we have %
\begin{equation}
\Phi_{\vecnull}^{(k\to j)}(\xi,\vecw,\vecz) = \Phi_{\vecalf_1,\ldots,\vecalf_N}^{(j)}(\xi,\vecw,\vecz) ,
\end{equation}
the transition probability density considered in the previous section.
In particular, for $k=j$,
\begin{equation}
\Phi_{\vecnull}^{(j\to j)}(\xi,\vecw,\vecz) = \nbar_j
\Phi_{\vecnull}(\nbar_j\xi,\vecw,\vecz) \prod_{\substack{i=1 \\ i\neq j}}^N
\int_{\nbar_i \xi}^\infty \Phi(\xi')\, d\xi' ,
\end{equation}
and for $k\neq j$,
\begin{equation}
\Phi_{\vecnull}^{(k\to j)}(\xi,\vecw,\vecz) = \nbar_j\,
\Phi(\nbar_j\xi,\vecw)\, \Phi(\nbar_k\xi,\vecz) \prod_{\substack{i=1 \\ i\neq j,k}}^N
\int_{\nbar_i \xi}^\infty \Phi(\xi')\, d\xi' ,
\end{equation}
since 
\begin{equation}
\Phi(\xi) = \int_{\scrB_1^{d-1}} \Phi(\xi,\vecw)\, d\vecw 
\end{equation}
and, by \cite[Eq.~(6.67)]{partII},
\begin{equation}
\Phi(\xi,\vecw) =\int_{\xi}^\infty \int_{\scrB_1^{d-1}} \Phi_{\vecnull}(\xi',\vecw,\vecz)\, d\vecz\,d\xi' .
\end{equation}
We note that $\Phi_{\vecnull}^{(k\to j)}(\xi,\vecw,\vecz)$ is a conditional probability density, in the sense that
\eqref{conditor} holds.
Note also that 
\begin{equation}
\nbar_k \Phi_{\vecnull}^{(k\to j)}(\xi,\vecw,\vecz)=\nbar_j \Phi_{\vecnull}^{(j\to k)}(\xi,\vecz,\vecw)
\end{equation}
since $\Phi_{\vecnull}(\xi,\vecw,\vecz)=\Phi_{\vecnull}(\xi,\vecz,\vecw)$.
We have the following relation with the density \eqref{Phi0P},
\begin{equation}
	\Phi_\vecnull^{(i)}(\xi,\vecz) = \sum_{j=1}^N \int_{\scrB_1^{d-1}} \Phi_{\vecnull}^{(i\to j)}(\xi,\vecw,\vecz)\, d\vecw .
\end{equation}

A second important case is when $\vecq$ is generic with respect to all lattices, i.e., $\vecalf_{i}\notin\QQ^d$ for all $i=1,\ldots,N$. In this case $\Phi_{\vecalf_1,\ldots,\vecalf_N}^{(j)}(\xi,\vecw,\vecz)$ is also independent of $\vecz$; we set
\begin{equation}
\Phi^{(j)}(\xi,\vecw) := \Phi_{\vecalf_1,\ldots,\vecalf_N}^{(j)}(\xi,\vecw,\vecz) ,
\end{equation}
and have the explicit representation
\begin{equation}
\Phi^{(j)}(\xi,\vecw)  =
\nbar_j
\Phi(\nbar_j\xi,\vecw) \prod_{\substack{i=1 \\ i\neq j}}^N
\int_{\nbar_i \xi}^\infty \Phi(\xi')\, d\xi' .
\end{equation}
This implies for instance
\begin{equation}\label{transi}
\begin{split}
	\Phi^{(j)}(\xi,\vecw) & = \sum_{i=1}^N \nbar_i\int_\xi^\infty \int_{\scrB_1^{d-1}} \Phi_{\vecnull}^{(i\to j)}(\xi',\vecw,\vecz)\, d\vecz\, d\xi' \\
& =  \nbar_j \sum_{i=1}^N \int_\xi^\infty \int_{\scrB_1^{d-1}} \Phi_{\vecnull}^{(j\to i)}(\xi',\vecz,\vecw)\, d\vecz\, d\xi' \\
	& = \nbar_j  \int_\xi^\infty \Phi_\vecnull^{(j)}(\xi',\vecw)\, d\xi' ,
\end{split}
\end{equation}
and hence in particular $\Phi^{(j)}(0,\vecw) = \nbar_j$ from \eqref{conditor}. 
A simple substitution shows that the stationary solution of \eqref{glB22} can be written as
\begin{equation}\label{pj}
	p^{(j)}(\vecV,\xi,\vecV_+) 
	=\sigma(\vecV,\vecV_+)\, \Phi^{(j)}\big(\xi,\vecb(\vecV,\vecV_+)\big) .
\end{equation}
Let us now discuss the key ingredient in the proof of Theorem \ref{exactpos1}.

\section{Equidistribution in products}\label{sec:Equidistribution}

Consider the subgroup $\widehat\Gamma=\Gamma_1\times\cdots\times\Gamma_N$ in $\SLR^N$,
where each $\Gamma_i$ is a lattice in $\SLR$.
We denote by $\mu_{\widehat\Gamma}$ the unique $\SLR^N$ invariant probability measure on $\widehat\Gamma\bs \SLR^N$, and by $\varphi$ the diagonal embedding of $\SLR$ in $\SLR^N$, i.e. $\varphi(M)=(M,\ldots,M)$. 
Let us set
\begin{equation} \label{NMINUSDEF}
	n_-(\vecx)=\begin{pmatrix} 1 & \vecx \\ \trans\vecnull & 1_{d-1} \end{pmatrix} \in \SLR
\end{equation}
and
\begin{equation} \label{PHITDEF}
	\Phi^t = \begin{pmatrix} \e^{-(d-1) t} & \vecnull \\ \trans\vecnull & \e^{t}1_{d-1} \end{pmatrix} \in \SLR.
\end{equation}
Recall that two lattices $\Gamma$ and $\Gamma'$ in $\SL(d,\R)$ are said to be \textit{commensurable} if
their intersection $\Gamma\cap\Gamma'$ is also a lattice;
otherwise $\Gamma$ and $\Gamma'$ are \textit{incommensurable}.

\begin{thm}\label{equiThm2}
Let $\Gamma_1,\ldots,\Gamma_N$ be pairwise incommensurable lattices in $\SLR$. Let $\lambda$ be a Borel probability measure on $\RR^{d-1}$, absolutely continuous with respect to Lebesgue measure, and let $f:\RR^{d-1}\times \widehat\Gamma\bs \SLR^N\to\RR$ be bounded continuous. Then
\begin{equation}\label{limiteq}
	\lim_{t\to\infty} \int_{\RR^{d-1}} f\big(\vecx,\varphi(n_-(\vecx)\Phi^{t})\big)\, d\lambda(\vecx)  = \int_{\RR^{d-1}\times\widehat\Gamma\bs \SLR^N} f(\vecx,g) \, d\lambda(\vecx) \, d\mu_{\widehat\Gamma}(g) .
\end{equation}
\end{thm}

\begin{proof}
The statement is classical for $N=1$ (cf.~\cite{partI}) and proved in \cite{farey} for $N=2$. The extension to general $N$ follows from the strategy in \cite{farey}: Ratner's measure classification theorem \cite{Ratner91a,Ratner91b} implies that there exists a closed connected subgroup $K\leq \SLR^N$ with the property that (a) $\widehat\Gamma\cap K$ is a lattice in $K$, (b) $H:=\varphi(\SLR)\leq K$ and (c) $\widehat\Gamma\bs\widehat\Gamma H$ is dense in $\widehat\Gamma\bs \widehat\Gamma K$.  
Shah's theorem \cite{Shah96} then shows that (for $\vecx$-independent test functions) the limit \eqref{limiteq} exists and is given by 
\begin{equation}
\int_{\widehat\Gamma\bs \widehat\Gamma K} f(g) \, d\mu_K(g)
\end{equation}
where $\mu_K(g)$ is the unique $K$-invariant probability measure on $\widehat\Gamma\bs \widehat\Gamma K$. The proof of \cite[Theorem 2]{farey} shows that the incommensurability of $\Gamma_i$ and $\Gamma_j$ implies that the projection of $K$ onto the $i$th and $j$th factor is the group $\SLR\times\SLR$. Lemma \ref{TWOTOMANYLEM} below shows that the only $K$ with this property is $K=\SLR^N$. The case of $\vecx$-dependent test functions $f$ follows from the same argument as in the proof of Theorem 5.3 in \cite{partI}.
\end{proof}

\begin{lem}\label{TWOTOMANYLEM} 
Let $K$ be a connected Lie subgroup of $\SL(d,\R)^N$ whose projection onto the
$i$th and $j$th factor equals $\SL(d,\R)\times\SL(d,\R)$, for any $i<j$.
Then $K=\SL(d,\R)^N$.
\end{lem}
\begin{proof}
Let $\ik$ be the Lie subalgebra of $\lsl(d,\R)^N$ corresponding to $K$.
The assumption implies that the projection of $\ik$ onto the $i$th and $j$th factor equals $\lsl(d,\R)\times\lsl(d,\R)$, 
for any $i<j$.
Since $\lsl(d,\R)$ is simple, there is a sequence $X_1,\ldots,X_{N-1}\in\lsl(d,\R)$ such that
$X:=[\ldots[[X_1,X_2],X_3],\ldots,X_{N-1}]\neq0$.
Let $p_j:\lsl(d,\R)^N\to\lsl(d,\R)$ denote projection onto the $j$th factor.
For each $j=1,\ldots,N-1$ we may now choose $Y_j\in\ik$ such that $p_1(Y_j)=X_j$ and $p_{j+1}(Y_j)=0$.
Then $[\ldots[[Y_1,Y_2],Y_3],\ldots,Y_{N-1}]=(X,0,\ldots,0)$.
It follows that if we let $\il$ be the set of those $Z\in\lsl(d,\R)$ for which $(Z,0,\ldots,0)\in\ik$,
then $\il\neq\{0\}$. Note also that our assumption (applied e.g.\ with $i=1$, $j=2$) implies that $\il$ is an ideal
in $\lsl(d,\R)$. Hence $\il=\lsl(d,\R)$, %
viz.\ $\lsl(d,\R)\times\{0\}\times\ldots\times\{0\}\subset\ik$.
Similarly $\{0\}\times\ldots\times\lsl(d,\R)\times\ldots\times\{0\}\subset\ik$,
with $\lsl(d,\R)$ in arbitrary position, and thus $\ik=\lsl(d,\R)^N$.
\end{proof}

The equidistribution of horospherical averages in Theorem \ref{equiThm2} implies, by the same argument as in \cite{partI}, the following equidistribution of spherical averages. The rotation $K(\vecv)\in\SO(d)$ is defined as in Section \ref{sec:Free}.

\begin{thm}\label{equiThm3}
Let $\Gamma_1,\ldots,\Gamma_N$ be pairwise incommensurable lattices in $\SLR$. Let $\lambda$ be a Borel probability measure on $\S_1^{d-1}$, absolutely continuous with respect to Lebesgue measure, and let $f:\S_1^{d-1}\times \widehat\Gamma\bs \SLR^N\to\RR$ be bounded continuous. Then
\begin{equation}\label{limiteq3}
	\lim_{t\to\infty} \int_{\S_1^{d-1}} f\big(\vecv,\varphi(K(\vecv)\Phi^{t})\big)\, d\lambda(\vecv)  = \int_{\S_1^{d-1}\times\widehat\Gamma\bs \SLR^N} f(\vecv,g) \, d\lambda(\vecv) \, d\mu_{\widehat\Gamma}(g) .
\end{equation}
\end{thm}

For our application to the Lorentz gas we are interested in the choice of lattices $\Gamma_i=M_i^{-1} \SLZ M_i$. The following lemma is a restatement of the fact that $\scrS$ is the commensurator of $\SLZ$ in $\SLR$.

\begin{lem}\label{littlelem}
The lattices $\Gamma_i=M_i^{-1} \SLZ M_i$ and $\Gamma_j=M_j^{-1} \SLZ M_j$ are commensurable if and only if $M_i M_j^{-1}\in\scrS$.
\end{lem}

In view of Lemma \ref{littlelem}, Theorem \ref{equiThm3} implies the following. We set $X_1:=\SLZ\backslash\SLR$ and denote by $\mu_1$ the unique $\SLR$ invariant probability measure on $X_1$.
Recall that we say that
{\em $M_1,\ldots,M_N\in\SLR$ are pairwise incommensurable} if $M_i M_j^{-1}\notin\scrS$ for all $i\neq j$.

\begin{thm}\label{SLDZTHM}
Assume $M_1,\ldots,M_N\in\SLR$ are pairwise incommensurable. Let $\lambda$ be a Borel probability measure on $\S_1^{d-1}$ which is absolutely continuous with respect to Lebesgue measure, and let $f:X_1^N\to\RR$ be bounded continuous. Then
\begin{multline}
	\lim_{t\to\infty} \int_{\S_1^{d-1}} f\bigl(M_1K(\vecv) \Phi^t,\ldots,M_N K(\vecv) \Phi^t\bigr)\, d\lambda(\vecv) \\
	= \int_{X_1^N} f(g_1,\ldots,g_N) \, d\mu_1(g_1) \cdots d\mu_1(g_N) .
\end{multline}
\end{thm}

As in \cite{partI}, the above equidistribution theorems can be extended to the semi-direct product group $\ASLR=\SLR\ltimes \RR^d$ with multiplication law
\begin{equation} \label{ASLMULTLAW}
	(M,\vecxi)(M',\vecxi')=(MM',\vecxi M' +\vecxi') .
\end{equation}
An action of $\ASLR$ on $\RR^d$ can be defined as
\begin{equation}
	\vecy \mapsto \vecy(M,\vecxi):=\vecy M+\vecxi .
\end{equation}
Each \textit{affine} lattice (i.e.\ translate of a lattice) 
of covolume one in $\RR^d$ can then be expressed as
$\ZZ^d g$ for some $g\in\ASLR$, 
and the space of affine lattices is then represented by $X=\ASLASL$ where $\ASLZ=\SLZ\ltimes \ZZ^d$.
We denote by $\mu$ the unique $\ASLR$ invariant probability measure on $X$. 

If $\vecalf\in\QQ^d$, say $\vecalf=\vecp/q$ for $\vecp\in\ZZ^d$, 
$q\in\Z_{>0}$, we see that
\begin{equation}
	\bigg(\ZZ^d + \frac{\vecp}{q}\bigg) \gamma M = \bigg(\ZZ^d + \frac{\vecp}{q}\bigg) M 
\end{equation}
for all
\begin{equation}\label{Gamq}
	\gamma \in \Gamma(q):= \{ \gamma\in\SLZ \col \gamma \equiv 1_d \bmod q \},
\end{equation}
the principal congruence subgroup. This means that the space of affine lattices with $\vecalf=\vecp/q$ can be parametrized by the homogeneous space $\Gamma(q)\backslash\SLR$ (this is not necessarily one-to-one). 

Given arbitrary $\vecalf_1,\ldots,\vecalf_N$ in $\R^d$, we introduce for each $i=1,\ldots,N$ a space $(X_i,\mu_i)$ 
and a map $\psi_i:\SL(d,\R)\to X_i$ as follows:
If $\vecalf_i\notin\Q^d$ then set $(X_i,\mu_i)=(X,\mu)$ and let $\psi_i$ be the map
$M\mapsto \ASL(d,\Z)(1_d,\vecalf_i)(M,\bn)$.
If $\vecalf_i\in\Q^d$ then fix some $q_i\in\Z^+$ so that $\vecalf_i\in q_i^{-1}\Z^d$,
set $X_i=\Gamma(q_i)\backslash\SLR$ with $\mu_i$ the unique $\SL(d,\R)$ invariant probability measure on $X_i$,
and let $\psi_i$ be the map $M\mapsto\Gamma(q_i)M$.
We now have the following generalization of Theorem \ref{SLDZTHM}.

\begin{thm}\label{equiThm9}
Let $\vecalf_1,\ldots,\vecalf_N\in\R^d$ be given and let $(X_i,\mu_i)$ and $\psi_i$ be as defined above.
Then for any $M_1,\ldots,M_N\in\SLR$ which are pairwise incommensurable,
any Borel probability measure $\lambda$ on $\S_1^{d-1}$ which is absolutely continuous with respect to 
Lebesgue measure, and any bounded continuous $f:X_{1}\times\cdots\times X_{N}\to\RR$,
\begin{multline}
	\lim_{t\to\infty} \int_{\S_1^{d-1}} f\bigl(\psi_1(M_1K(\vecv)\Phi^t),\ldots,\psi_N(M_N K(\vecv)\Phi^t)\bigr)\, d\lambda(\vecv) \\
	= \int_{X_1\times\cdots\times X_N} f(g_1,\ldots,g_N) \, d\mu_{1}(g_1) \cdots d\mu_{N}(g_N) .
\end{multline}
\end{thm}

\begin{proof}
Let $\widehat G=G_1\times\cdots\times G_N$ and $\widehat \Gamma=\Gamma_1\times\cdots\times\Gamma_N$, where
\begin{align}
G_i=\begin{cases}\ASL(d,\R)&\text{if }\:\vecalf_i\notin\Q^d\\
\SL(d,\R)&\text{if }\:\vecalf_i\in\Q^d\end{cases}
\end{align}
and
\begin{align}
\Gamma_i=\begin{cases}
\bigl((1_d,\vecalf_i)M_i\bigr)^{-1}\ASL(d,\Z)\bigl((1_d,\vecalf_i)M_i\bigr)&\text{if }\:\vecalf_i\notin\Q^d
\\
M_i^{-1}\Gamma(q_i)M_i&\text{if }\:\vecalf_i\in\Q^d. %
\end{cases}
\end{align}
We here view $\SL(d,\R)$ as a subgroup of $\ASL(d,\R)$ through $M\mapsto(M,\bn)$.
Let $\varphi:\SL(d,\R)\to\widehat G$ be the diagonal imbedding.
Again by Ratner \cite{Ratner91a,Ratner91b}, there exists a closed connected subgroup $K\leq\widehat G$ such that
$\widehat\Gamma\cap K$ is a lattice in $K$,
$H:=\varphi(\SL(d,\R))\leq K$, and $\widehat\Gamma K$ equals the closure of $\widehat\Gamma H$ in $\widehat G$.
We are going to prove that $K=\widehat G$.
Once we have this, the proof of %
Theorem \ref{SLDZTHM} extends immediately to the present situation,
thus completing the proof of Theorem \ref{equiThm9}.

Let $\rho:\ASL(d,\R)\to\SL(d,\R)$ be the projection $(M,\vecxi)\mapsto M$,
and define $\widehat\rho:\widehat G\to\SL(d,\R)^N$
through $\widehat\rho(g_1,\ldots,g_N)=(g_1',\ldots,g_N')$ where
$g_i'=\rho(g_i)$ if $\vecalf_i\notin\Q^d$ and
$g_i'=g_i$ if $\vecalf_i\in\Q^d$.
Then 
\begin{align}
\widehat\rho(\widehat\Gamma)=\Gamma_1'\times\cdots\times\Gamma_N'
\quad\text{where }\:
\Gamma_i'\:\:\begin{cases}=M_i^{-1} \SL(d,\Z) M_i
&\text{if }\:\vecalf_i\notin\Q^d
\\
=\Gamma_i=M_i^{-1}\Gamma(q_i)M_i&\text{if }\:\vecalf_i\in\Q^d. %
\end{cases}
\end{align}
Note that $\Gamma_1',\ldots,\Gamma_N'$ are pairwise incommensurable,
by Lemma \ref{littlelem}.
Hence by the proof of Theorem \ref{equiThm2},
$\widehat\rho(\widehat\Gamma H)$ is dense in $\SL(d,\R)^N$.
However $\widehat\rho(\widehat\Gamma H)\subset\widehat\rho(\widehat\Gamma K)$,
and $\widehat\rho(\widehat\Gamma K)$  %
is closed in
$\SL(d,\R)^N$, %
since $\widehat\rho$ induces a projection map
$\widehat\Gamma\backslash\widehat G\to\widehat\rho(\widehat\Gamma)\backslash\SL(d,\R)^N$
with compact fibers.
Therefore $\widehat\rho(\widehat\Gamma K)=\SL(d,\R)^N$ and so
\begin{align}\label{equiThm9pf1}
\widehat\rho(K)=\SL(d,\R)^N.
\end{align}

Next, let $p_i:\widehat G\to G_i$ be the projection onto the $i$th factor.
If $\vecalf_i\in\Q^d$ then $p_i(K)=\SL(d,\R)$, since $H\leq K$ and $p_i(H)=\SL(d,\R)$.

Now assume $\vecalf_i\notin\Q^d$.
Then %
$p_i(\widehat\Gamma)=\Gamma_i=\bigl((1_d,\vecalf_i)M_i\bigr)^{-1}\ASL(d,\Z)\bigl((1_d,\vecalf_i)M_i\bigr)$,
and $\Gamma_i\SL(d,\R)$ is dense in $G_i=\ASL(d,\R)$ by \cite[proof of Thm.\ 5.2]{partI}.
Furthermore, since $\widehat\Gamma\cap K$ is a lattice in $K$,
$p_i(\widehat\Gamma\cap K)$ is a lattice in $p_i(K)$
(cf.\ \cite[Lemma 1.6]{Raghunathan});
hence also $\Gamma_i\cap p_i(K)$ is a lattice in $p_i(K)$ and $p_i(K)$ is a closed subgroup of $G_i$,
so that $\Gamma_i p_i(K)$ is closed in $G_i$.
But $\Gamma_i\SL(d,\R)=\Gamma_ip_i(H)\subset\Gamma_i p_i(K)$. Hence
\begin{align}\label{equiThm9pf2}
p_i(K)=G_i,\qquad i=1,\ldots,N.
\end{align}

We will show that \eqref{equiThm9pf1} and \eqref{equiThm9pf2} together imply $K=\widehat G$.
Let $\ig_i$ be the Lie algebra of $G_i$; then $\widehat\ig=\ig_1\oplus\cdots\oplus\ig_N$ is the Lie algebra of $\widehat G$.
Let $\ik$ be the Lie subalgebra of $\widehat\ig$ corresponding to $K$.
Let us first assume $\vecalf_1\notin\Q^d$;
then $\ig_1=\asl(d,\R)$, which we may identify in a natural way with the linear space $\lsl(d,\R)\oplus\R^d$
endowed with the Lie bracket $[(A,\vecv),(B,\vecw)]=([A,B],\vecv B-\vecw A)$
(cf., e.g., \cite[Prop.\ 1.124]{Knapp}).
Set 
\begin{align}
\il=\{X\in\asl(d,\R)\col (X,0,\ldots,0)\in\ik\}.
\end{align}
Using $dp_1(\ik)=\ig_1$ (cf.\ \eqref{equiThm9pf2}) it follows that $\il$ is an ideal in $\ig_1$;
hence also $d\rho(\il)$ is an ideal in $\lsl(d,\R)$
(note that $d\rho:\asl(d,\R)\to\lsl(d,\R)$ is the map $(A,\vecv)\mapsto A$).
It follows from \eqref{equiThm9pf1} that $d\widehat\rho(\ik)=\lsl(d,\R)^N$;
hence for any given $A,B\in\lsl(d,\R)$, there exist
$\vecv,\vecw\in\R^d$ and $X_i,Y_i\in\ig_i$ ($i=2,\ldots,N$) such that
$X_i=Y_i=0$ when %
$\vecalf_i\in\Q^d$,
$d\rho(X_i)=d\rho(Y_i)=0$ when %
$\vecalf_i\notin\Q^d$, and
\begin{align}
\bigl((A,\vecv),X_2,\ldots,X_N\bigr),\:
\bigl((B,\vecw),Y_2,\ldots,Y_N\bigr)\in\ik.
\end{align}
Hence also
\begin{align}
\Bigl[\bigl((A,\vecv),X_2,\ldots,X_N\bigr),\bigl((B,\vecw),Y_2,\ldots,Y_N\bigr)\Bigr]
=\bigl(([A,B],\vecv B-\vecw A),0,\ldots,0\bigr)\in\ik.
\end{align}
Therefore $d\rho(\il)$ contains $[A,B]$ for any $A,B\in\lsl(d,\R)$;
since $\lsl(d,\R)$ is simple this implies $d\rho(\il)=\lsl(d,\R)$.
Now fix some $A\in\lsl(d,\R)$ which is invertible as a $d\times d$ matrix.
Because of $d\rho(\il)=\lsl(d,\R)$ there is some $\vecv\in\R^d$ such that $((A,\vecv),0,\ldots,0)\in\ik$.
Using also $dp_1(\ik)=\asl(d,\R)$ we see that for any $\vecw\in\R^d$ there exist some
$X_j\in\ig_j$ for $j=2,\ldots,N$ such that
$((A,\vecw),X_2,\ldots,X_N)\in\ik$.
Hence $\ik$ also contains their Lie product, viz.
\begin{align}
\Bigl((0,(\vecv-\vecw)A),0,\ldots,0\Bigr)\in\ik.
\end{align}
Hence, since $A$ is invertible and $\vecw$ is arbitrary, $\il$ contains $(0,\vecx)$ for all $\vecx\in\R^d$.
Together with $d\rho(\il)=\lsl(d,\R)$ this implies $\il=\ig_1$, and therefore
\begin{align}\label{equiThm9pf3}
\ig_1\times\{0\}\times\cdots\times\{0\}\subset\ik.
\end{align}

We proved \eqref{equiThm9pf3} under the assumption that $\vecalf_1\notin\Q^d$;
however the same argument in a simplified form applies when $\vecalf_1\in\Q^d$, i.e.\
\eqref{equiThm9pf3} holds in that case as well.
By analogous reasoning we get $\{0\}\times\cdots\times\ig_i\times\cdots\times\{0\}\subset\ik$ for each $i$.
Hence $\ik=\widehat\ig$ and $K=\widehat G$, and the proof is complete.
\end{proof}

Finally, Theorems \ref{freeThm} and \ref{exactpos1} now follow from Theorem \ref{equiThm9} by the same steps as in \cite[Sections 6 and 9]{partI}.

\end{document}